\documentclass[envcountsect]{llncs}
\pdfoutput=1
\usepackage[colorlinks,urlcolor=black]{hyperref}
\usepackage{mdwlist} 
\usepackage[normalem]{ulem}  

\usepackage{amssymb,stmaryrd,amsmath,xspace}
\usepackage{nccmath}
\usepackage{prooftree}
\usepackage{graphicx}
\usepackage[colorlinks,urlcolor=black]{hyperref}
\usepackage{soul}
\usepackage[all]{xy}

\usepackage{etex} 
\usepackage{tikz}
\definecolor{shade}{HTML}{F6F6FF}

 \renewenvironment{thebibliography}[1]{%
   \begin{oldthebibliography}{#1}%
     \setlength{\parskip}{0ex}%
     \setlength{\itemsep}{0ex}%
     \fontsize{8.4}{9.4}
     \selectfont
 }%
 {%
   \end{oldthebibliography}%
 }

\usepackage{float}
\floatstyle{boxed} \restylefloat{figure}

\usepackage{fancybox}
\newlength{\mylength}
\newenvironment{frameqn}%
{\setlength{\fboxsep}{4pt}
\setlength{\mylength}{\linewidth}%
\addtolength{\mylength}{-2\fboxsep}%
\addtolength{\mylength}{-2\fboxrule}%
\Sbox
\minipage{\mylength}%
$$}%
{$$\endminipage\endSbox
{
\[\fbox{\TheSbox}\]
}}

\newenvironment{frametxt}%
{
\setlength{\fboxsep}{4pt}
\setlength{\mylength}{\linewidth}%
\addtolength{\mylength}{-2\fboxsep}%
\addtolength{\mylength}{-2\fboxrule}%
\Sbox
\minipage{\mylength}%
}%
{\endminipage\endSbox
{
\[\fbox{\TheSbox}\]
}}

{
\begin{figure}[H]
}
{\end{figure}
}

\newcommand\hnu{\scalebox{.8}{$\new$}} 

\newcommand\fa{\f{fa}}

\newcommand\ns[1]{\mathsf{#1}} 

\newcommand\supp{\f{supp}}

\newcommand{\denot}[3]{\llbracket #3 \rrbracket_{\scalebox{.6}{$#2$}}^{\hspace{-.1ex}\scalebox{.4}{$#1$}}}

\newbox\tempa
\newbox\tempb
\newdimen\tempc
\def\mud#1{\hfil $\displaystyle{\mathstrut #1}$\hfil}
\def\rig#1{\hfil $\displaystyle{#1}$}
\def\irulehelp#1#2#3{\setbox\tempa=\hbox{$\displaystyle{\mathstrut #2}$}%
		        \setbox\tempb=\vbox{\halign{##\cr
	\mud{#1}\cr
	\noalign{\vskip\the\lineskip}%
	\noalign{\hrule height 0pt}%
	\rig{\vbox to 0pt{\vss\hbox to 0pt{${\; #3}$\hss}\vss}}\cr
	\noalign{\hrule}%
	\noalign{\vskip\the\lineskip}%
	\mud{\copy\tempa}\cr}}%
		      \tempc=\wd\tempb
		      \advance\tempc by \wd\tempa
		      \divide\tempc by 2 }
\def\irule#1#2#3{{\irulehelp{#1}{#2}{#3}%
		     \hbox to \wd\tempa{\hss \box\tempb \hss}}}

\newcommand\den[1]{{\hspace{-.1ex}\scalebox{.5}{$#1$}}}

\newcommand\iden{\den{\interp I}} 
 
\newcommand\nden{\den{\mathcal N}}
\newcommand\mden{\den{\interp M}}
\newcommand{\ndenot}[2]{\denot{\mathcal N}{#1}{#2}}

\usepackage[mathscr]{eucal}
\newcommand\interp[1]{\ensuremath{\mathscr #1}}

\newcommand{\freshwedge}[1]{\mbox{$\bigwedge^{\hspace{-.5ex}\raisebox{-.2ex}{\scalebox{.6}{$\# #1$}}}$}}
\newcommand{\freshvee}[1]{\mbox{$\bigvee^{\hspace{-.2ex}\raisebox{-.2ex}{\scalebox{.6}{$\# #1$}}}$}}
\newcommand{\model}[1]{\denot{\interp M}{}{#1}}

\newcommand\deffont[1]{{\bf #1}}
\usepackage{jamiebsy}
\newcommand\tf[1]{{\jpmb{\mathsf{#1}}}}
\newcommand\f[1]{\mathit{#1}}

\newcommand\act{{\cdot}}

\newcommand\Func{{\Rightarrow}}
\newcommand\limp{\Rightarrow}
\newcommand\Forall[1]{\forall #1.}
\newcommand\Exists[1]{\exists #1.}

\newcommand\id{\f{id}}

\newcommand\cent{\vdash}
\newcommand\ment{\vDash}
\newcommand\sm{{\mapsto}}
\newcommand\ssm{{{:}{=}}}
\newcommand\mone{{\text{-}1}}
\newcommand\rulefont[1]{\ensuremath{\bf (#1)}\xspace}
\newcommand\new{\reflectbox{\ensuremath{\mathsf{N}}}}

\newcommand\atoms{{\mathbb A}}

\newcommand\tneg{{\pmb\neg}}

\newcommand\tbot{{\pmb\bot}}
\newcommand\tand{{\pmb\wedge}}

\newcommand\tall{{\pmb\forall}}

\newcommand\lneg{{\neg}}
\newcommand\lall[1]{\mbox{$\scalebox{.8}{\raisebox{.1ex}{$\bigwedge$}}^{\hspace{-.5ex}\raisebox{-.2ex}{\scalebox{.6}{$\# #1$}}}$}}
\newcommand\lexi[1]{\mbox{$\scalebox{.8}{\raisebox{.1ex}{$\bigvee$}}^{\hspace{-.2ex}\raisebox{-.2ex}{\scalebox{.6}{$\# #1$}}}$}}
\newcommand\ltop{\top}
\newcommand\lbot{\bot}
\renewcommand\land{\wedge}
\renewcommand\lor{\vee}

\spnewtheorem{thrm}{Theorem}[section]{\bfseries}{\itshape}
\spnewtheorem{lemm}[thrm]{Lemma}{\bfseries}{\itshape}
\spnewtheorem{prop}[thrm]{Proposition}{\bfseries}{\itshape}
\spnewtheorem{corr}[thrm]{Corollary}{\bfseries}{\itshape}

\spnewtheorem{nttn}[thrm]{Notation}{\bfseries}{}
\spnewtheorem{defn}[thrm]{Definition}{\bfseries}{}
\spnewtheorem{xmpl}[thrm]{Example}{\bfseries}{}
\spnewtheorem{rmrk}[thrm]{Remark}{\bfseries}{}

\newenvironment{frcseries}{\fontfamily{frc}\selectfont}{}

\DeclareFontFamily{OT1}{pzc}{}
\DeclareFontShape{OT1}{pzc}{m}{it}%
              {<-> s * [1.100] pzcmi7t}{}
\DeclareMathAlphabet{\mathpzc}{OT1}{pzc}%
                                 {m}{it}

\usepackage{array}
\newcolumntype{L}[1]{>{$}p{#1}<{$}}
\newcolumntype{C}[1]{>{\centering$}p{#1}<{$}}
\newcolumntype{R}[1]{>{\raggedleft$}p{#1}<{$}}
\newcommand\maketab[2]
   {\newenvironment{#1}
      {\begin{quote}\noindent\begin{tabular}{#2}}
      {\end{tabular}\end{quote}}
    \newenvironment{#1noquote}{\noindent\begin{tabular}{#2}}{\end{tabular}}
    }

\usepackage{fancyhdr}
\begin{document}
\setlength{\abovedisplayskip}{2pt}%
\setlength{\belowdisplayskip}{2pt}%

\title{Nominal semantics for predicate logic: \\ algebras, substitution, quantifiers, and limits\thanks{%
\tiny{\textbf{This is the author's version.  \href{https://ceur-ws.org/Vol-857/paper_f08.pdf}{Version of Record} (\href{https://web.archive.org/web/20231002232749/https://ceur-ws.org/Vol-857/paper_f08.pdf}{permalink}) was published in pages 104--118 of the \emph{Proceedings of the 9th Italian Convention on Computational Logic} (CILC 2012), edited by Francesca A. Lisi, CEUR Workshop Proceedings Volume 857, ISSN 1613-0073 (\href{https://nbn-resolving.org/urn:nbn:de:0074-857-8}{urn:nbn:de:0074-857-8}), and should be so cited. Copyright is with the authors. 
}}}}
\author{Gilles Dowek \and Murdoch J. Gabbay}
\authorrunning{Dowek and Gabbay}
\institute{}

\maketitle
\vspace{-1.5em}
\begin{abstract}
We define a model of predicate logic in which every term and predicate, open or closed, has an absolute denotation independently of a valuation of the variables.  
For each variable $a$, the domain of the model contains an element $\llbracket a \rrbracket$ which is the denotation of the term $a$ (which is also a variable symbol).
Similarly, the algebra interpreting predicates in the model directly interprets open predicates. 
Because of this models must also incorporate notions of substitution and quantification.

These notions are axiomatic, and need not be applied only to sets of syntax.
We prove soundness and show how every `ordinary' model (i.e. model based on sets and valuations) can be translated to one of our nominal models, and thus also prove completeness.
\vspace{-1ex}
\begin{keywords}
Lattices and algebra; First-order logic; Nominal semantics
\end{keywords}
\end{abstract}

\section{Introduction}

What is a notion of algebraic truth suitable for interpreting predicate logic?

For \emph{propositional} logic an answer is clear: Heyting algebra for the intuitionistic case, and Boolean algebra for the classical case.\footnote{Accessible references for Heyting and Boolean algebra are easily found online.  For a modern and encyclopaedic treatment, see e.g. \cite{vickers:topvl}.} 
This is not sufficient for interpreting predicate logic because predicates can contain variable symbols.

The standard method, due to Tarski \cite{tarski:semct}, uses valuations.
We assume a domain $D$ and a valuation $\varsigma$ mapping variables to $D$, and then terms and predicates are given denotation in the context of the valuation: $\denot{}{\varsigma}{r}$ and $\denot{}{\varsigma}{\phi}$.
Intuitively, valuations convert predicates to propositions, which reduces the problem of truth to the infinite propositional case.

Thus with the usual notion of model, names denote semantic objects, but semantic objects have no inherent notion of names.

We propose an alternative.
Allow semantic objects to be aware of names.
Then, we can map a variable $a$ \emph{to a representation of itself} in the denotation. 
Thus, open terms and open predicates map to open elements of a domain and algebra of truth-values, like so: $\denot{}{}{r}$ and $\denot{}{}{\phi}$.
Connectives, including quantifiers, just become operations on the domain. 
So, the denotation of $a$ is $\denot{}{}{a}$ and the denotation of $\tall a.\phi$ is some operation $\lall a$ applied to $\denot{}{}{\phi}$.

What it means for an element to be open is given for us by the recently-developed theory of \emph{nominal semantics}, which was designed to provide semantics for variable symbols in inductive syntax-with-binding \cite{gabbay:newaas-jv,gabbay:fountl}.
In this paper we apply nominal ideas to semantic objects which, unlike the syntax of terms and predicates, need not be inductively defined.
We use \emph{nominal algebra} to axiomatise substitution following \cite{gabbay:capasn,gabbay:capasn-jv} and a simple theory of \emph{$A\#$fresh limits} in partial orders, which is introduced in this paper.
Using $A\#$fresh limits we show how $\ltop$, $\land$, and the quantifier $\lall{a}$ are all just greatest $A\#$lower fresh bounds for suitable finite sets.
Specifically:
\begin{itemize*}
\item
$\ltop$ is the greatest lower bound for the empty set $\varnothing$.
\item
$x\land y$ is the greatest lower bound for the set $\{x,y\}$, i.e. $\bigwedge\{x,y\}$. 
\item
$\lall{a}x$ is the greatest $\{a\}\#$lower bound for the set $\{x\}$, i.e. $\freshwedge{\{a\}}\{x\}$.
\end{itemize*} 
We use this to construct a nominal algebraic semantics for first-order logic and prove it sound and complete.
The completeness proof is a little unorthodox: instead of a direct proof (which is not hard) we give a more general lifting construction from Tarski-style valuation models, to our nominal models. 
This gives a direct translation of the familiar semantics into our new one and we obtain completeness as a corollary.

\subsubsection*{Map of the paper}

Section~\ref{sect.background} presents the technical background: nominal sets and first-order logic.
Section~\ref{sect.posets} introduces nominal posets and substitution algebras.
This material is partly new (nominal posets and $\freshwedge{A}$) and partly a restatement of \cite{gabbay:capasn,gabbay:capasn-jv} (substitution algebras; though to be fair, the presentation here is significantly simplified and updated).
Section~\ref{sect.model} introduces nominal Boolean algebras, which is a mathematically precise answer to the question with which we opened this Introduction, and proves soundness.
Finally, Section~\ref{sect.complete} presents the `lifting' operation on models and proves completeness.

\section{Technical background}
\label{sect.background}

\subsection{Background on nominal sets}

Nominal sets were introduced in \cite{gabbay:newaas-jv} (where they were called \emph{equivariant Fraenkel-Mostowski sets}).
The presentation below is quite self-contained; see \cite{gabbay:newaas-jv} or \cite{gabbay:fountl} for more details.

We want names to `exist' in the denotation.
We do this by endowing sets with a symmetry group action over name permutations (Definition~\ref{defn.fin.supp}). 
We extract from this the key idea of \emph{support} (Definition~\ref{defn.support}), which captures the degree of asymmetry of an element under the group action.
We use this to define a \emph{nominal set} $\ns X$ (Definition~\ref{defn.nominal.set}) as a set with a group action all of whose elements $x$ are at most finitely asymmetric (have finite support), where $a\#x$ and $a\not\in\supp(x)$ mean the same thing: ``$x$ is symmetric over/does not depend on the name $a$''.
We do not \emph{a priori} assume e.g. a substitution action on names;
in our framework this is algebraically definable just from the more primitive notion of permutable names and finite asymmetry (see e.g. Figure~\ref{fig.nom.sigma}).

\begin{defn}
\label{defn.atoms}
Fix a countably infinite set of \deffont{atoms} $\mathbb A$.
We use a \deffont{permutative convention} that $a,b,c,\ldots$ range over \emph{distinct} atoms.
\end{defn}

\begin{defn}
A \deffont{(finite) permutation} $\pi$ is a bijection on the set $\atoms$ such that $\f{nontriv}(\pi)=\{a\mid \pi(a)\neq a\}$ is finite.

Write $\id$ for the \deffont{identity} permutation such that $\id(a)=a$ for all $a$.
Write $\pi'\circ\pi$ for composition, so that $(\pi'\circ\pi)(a)=\pi'(\pi(a))$.
Write $\pi^\mone$ for inverse, so that $\pi^\mone\circ\pi=\id=\pi\circ\pi^\mone$.
Write $(a\;b)$ for the \deffont{swapping} (terminology from \cite{gabbay:newaas-jv}) mapping $a$ to $b$,\ $b$ to $a$,\ and all other $c$ to themselves, and take $(a\;a)=\id$.
\end{defn}

\begin{defn}
\label{defn.fin.supp}
A \deffont{set with a permutation action} $\ns X$ is a pair $(|\ns X|,\act)$ of an \deffont{underlying set} $|\ns X|$ and a \deffont{permutation action} written $\pi\act x$ 
 which is a group action on $|\ns X|$, so that $\id\act x=x$ and $\pi\act(\pi'\act x)=(\pi\circ\pi')\act x$ for all $x\in|\ns X|$ and permutations $\pi$ and $\pi'$.
\end{defn}

\begin{defn}
\label{defn.support}
Say that $A\subseteq\mathbb A$ \deffont{supports} $x\in|\ns X|$ when for every $\pi$, if $\pi(a)=a$ for every $a\in A$ then $\pi\act x=x$.
If some finite $A$ supporting $x$ exists, then call $x$ \deffont{finitely-supported}.
Then define the \deffont{support} of $x$ by
$$
\supp(x)=\bigcap\{A\mid A\text{ finite},\ A\text{ supports }x\} .
$$
Write $a\#x$ as shorthand for $a\not\in\supp(x)$ and read this as $a$ is \deffont{fresh for} $x$.
\end{defn}

\begin{thrm}
\label{thrm.supp}
Suppose $\ns X$ is a set with a permutation action and $x\in|\ns X|$.
Then if $x$ has finite support then 
$\supp(x)$ supports $x$ and is the unique least finite supporting set of $x$.
\end{thrm}
\begin{proof}
See \cite[Theorem~2.21]{gabbay:fountl}.
A different but equivalent formulation of the result is in \cite{gabbay:newaas-jv}.
\end{proof}

\begin{frametxt}
\begin{defn}
\label{defn.nominal.set}
Suppose $\ns X$ is a set with a permutation action.

Call a set with a permutation action $\ns X$ a \deffont{nominal set} when every $x\in|\ns X|$ has finite support.
$\ns X$, $\ns Y$, $\ns Z$ will range over nominal sets.
\end{defn}
\end{frametxt}

\begin{xmpl}
\label{xmpl.pow}
\begin{itemize*}
\item
$\mathbb A$ is a nominal set where $\pi\act a=\pi(a)$.
It is easy to check that $\supp(a)=\{a\}$.
\item
The sets of terms and predicates from Definition~\ref{defn.terms} below, are nominal sets if we let permutation act in the natural way on the atoms (considered as variable symbols) inside them.
It is easy to check that $\supp(r)$ or $\supp(\phi)$ is equal to the atoms occurring in $r$ or $\phi$ if we \emph{do not} take terms up to $\alpha$-equivalence, and $\supp(r)$ or $\supp(\phi)$ is equal to the atoms occurring \emph{free} in $r$ or $\phi$ if we \emph{do} take $r$ or $\phi$ up to $\alpha$-equivalence.
A detailed study of this is in \cite[Section~5]{gabbay:fountl}; see also \cite{crole:alpee}.
\end{itemize*}
Later on in this paper we shall see more examples.
Notably, the sets $X^\bullet$ from Definition~\ref{defn.idenot.beta} are nominal sets, and the set of \emph{valuations} from Definition~\ref{defn.varsigma} is a set with a permutation action but not a nominal set (see Remark~\ref{rmrk.not.nominal} and associated footnote).
For more examples see \cite{gabbay:fountl}.
\end{xmpl}

\begin{defn}
\label{defn.NOM}
Call a function $f\in |\ns X|\to|\ns Y|$ \deffont{equivariant} when $\pi\act f(x)=f(\pi\act x)$ for all permutations $\pi$ and $x\in|\ns X|$.
In this case write $f:\ns X\to\ns Y$.
\end{defn}


\begin{prop}
\label{prop.pi.supp}
$\supp(\pi\act x)=\{\pi(a)\mid a\in\supp(x)\}$. 
\end{prop}
\begin{proof}
It is not hard to check that $A$ supports $x$ if and only if $\{\pi(a)\mid a\in A\}$ supports $\pi\act x$.
We use Theorem~\ref{thrm.supp}.
\end{proof}

\begin{corr}
\label{corr.stuff}
\begin{enumerate*}
\item
If $\pi(a)=a$ for all $a\in\supp(x)$ then $\pi\act x=x$.
\item
If $\pi(a)=\pi'(a)$ for every $a\in\supp(x)$ then $\pi\act x=\pi'\act x$.
\item
$a\#x$ if and only if $\Exists{b}b\#x\wedge (b\;a)\act x=x$.
\end{enumerate*}
\end{corr}
\begin{proof}
Parts~1 and~2 follow from the fact that $\supp(x)$ supports $x$ (Theorem~\ref{thrm.supp}).
Part~3 follows using Proposition~\ref{prop.pi.supp}.
\end{proof}

\subsection{First-order logic}

First-order logic is of course standard.
The reader can find any number of presentations in the literature, for instance \cite{dalen:logs}.

\begin{defn}
Here and for the rest of this paper fix a signature 
$$
\Sigma=(\f{TermFormers},\f{PredicateFormers},\f{arity})
$$ 
of \deffont{term-formers} and \deffont{predicate-formers} and an \deffont{arity} function $\f{arity}$ mapping $\f{TermFormers}\cup\f{PredicateFormers}$ to $\mathbb N=\{0,1,2,\dots\}$.
We let $\tf f$ range over distinct term-formers and $\tf P$ range over distinct predicate-formers.
\end{defn}

\begin{defn}
\label{defn.terms}
Define \deffont{terms} $r$ and \deffont{predicates} $\phi$ inductively by:
\begin{frameqn}
\begin{array}{r@{\ }l}
r::=& a \mid \tf f(r_1,\dots,r_{\f{arity}(\tf f)})
\\
\phi::=& \tbot \mid \tf P(r_1,\dots,r_{\f{arity}(\tf P)}) \mid \phi\tand \phi \mid \tneg \phi \mid \tall a.\phi
\end{array}
\end{frameqn}
Define \deffont{free atoms} inductively by:
$$
\begin{array}{r@{\ }l@{\quad}r@{\ }l@{\quad}r@{\ }l}
\fa(a)=&\{a\}
&
\fa(\tf f(r_1,\dots,r_n))=&\bigcup_i \fa(r_i)
\\
\fa(\tbot)=&\varnothing
&
\fa(\tf P(r_1,\dots,r_n))=&\bigcup_i \fa(r_i)
\\
\fa(\tneg\phi)=&\fa(\phi)
&
\fa(\phi_1\tand\phi_2)=&\fa(\phi_1)\cup\fa(\phi_2)
&
\fa(\tall a.\phi)=&\fa(\phi){\setminus}\{a\}
\end{array}
$$
\end{defn}

\begin{defn}
We take predicates $\phi$ up to $\alpha$-equivalence as usual (so for instance $\tall a.\tf P(a)=\tall b.\tf P(b)$ where $\f{arity}(\tf P)=1$).

We write $r[a\ssm s]$ and $\phi[a\ssm s]$ for the usual capture-avoiding substitution on terms and predicates.
For instance, $\tf f(a)[a\ssm b]=\tf f(b)$ and $(\tall b.\tf P(a))[a\ssm b]=\tall b'.\tf P(b)$.
\end{defn}

\begin{defn}
Let $\Phi$ and $\Psi$ range over finite sets of predicates.
A \deffont{sequent} is a pair $\Phi\cent\Psi$.
We define the \deffont{derivable sequents} as usual for classical first-order logic by the rules in Figure~\ref{fig.deriv}.
\end{defn} 

\begin{figure}[t]
\begin{minipage}{\textwidth}
$$
\begin{gathered}
\begin{prooftree}
\phantom{h}
\justifies
\Phi,\tbot\cent\Psi
\using\rulefont{\tbot L}
\end{prooftree}
\qquad
\begin{prooftree}
\phi_1,\phi_2,\Phi\cent\Psi
\justifies
\phi_1\tand\phi_2,\Phi\cent\Psi
\using\rulefont{\tand L}
\end{prooftree}
\qquad
\begin{prooftree}
\Phi\cent\psi_1,\Psi\quad
\Phi\cent\psi_2,\Psi
\justifies
\Phi\cent\psi_1\tand\psi_2,\Psi
\using\rulefont{\tand R}
\end{prooftree}
\\[1ex]
\begin{array}{c@{\qquad}c}
\begin{prooftree}
\Phi\cent\phi,\Psi
\justifies
\Phi,\tneg\phi\cent\Psi
\using\rulefont{\tneg L}
\end{prooftree}
&
\begin{prooftree}
\Phi,\phi\cent\Psi
\justifies
\Phi\cent\tneg\phi,\Psi
\using\rulefont{\tneg R}
\end{prooftree}
\\[4ex]
\begin{prooftree}
\Phi,\phi[a\ssm s]\cent\Psi
\justifies
\Phi,\tall a.\phi\cent\Psi
\using\rulefont{\tall L}
\end{prooftree}
&
\begin{prooftree}
\Phi\cent\psi,\Psi\quad (a\not\in\fa(\Phi,\Psi))
\justifies
\Phi\cent\tall a.\psi,\Psi
\using\rulefont{\tall R}
\end{prooftree}
\end{array}
\end{gathered}
$$
\end{minipage}
\caption{Derivability in first-order classical logic} 
\label{fig.deriv}
\end{figure}

\section{Partial orders and substitution}
\label{sect.posets}

We obtain our nominal semantics for first-order logic by combining two things: a partial order, and a substitution action.
We need the partial order to interpret entailment.
We need the substitution action to interpret variables and substitutions. 
These are Subsections~\ref{subsect.nominal.poset} and~\ref{subsect.sigma}.

Of course, these two things interact via quantifiers.
Thus we introduce the novel notion of $\freshwedge{A}X$, an $A\#$greatest lower bound for a set $X$ (Notation~\ref{nttn.Afresh.glb}).
It turns out that this generalises $\lbot$, $\land$, and $\lall{a}$ into a single definition.

\subsection{Nominal posets (nominal partially ordered set)}
\label{subsect.nominal.poset}

We can now begin to ask what happens if we combine the notion of partially ordered set with the notion of nominal set.
In particular, we can think about greatest lower and least upper bounds \emph{in the presence} of nominal freshness.
This leads us to the idea of $\freshwedge{A} X$ and $\freshvee{A} X$; greatest lower and least upper bounds amongst elements that do not include $A$ in their support. 

\begin{defn}
A \deffont{nominal poset} is a tuple $\mathcal L=(|\mathcal L|,\act,\leq)$ such that $(|\mathcal L|,\act)$ is a nominal set and $\leq\subseteq|\mathcal L|\times|\mathcal L|$ is an equivariant partial order.\footnote{So $x\leq y$ if and only if $\pi\act x\leq\pi\act y$.}

Call $\mathcal L$ \deffont{finitely fresh-complete} when for every finite subset $X\subseteq|\mathcal L|$ and every finite set of atoms $A\subseteq\mathbb A$ the set of $A$-fresh lower bounds
$$
\{x'\in|\mathcal L| \mid A\cap\supp(x')=\varnothing\ \wedge\ \Forall{x{\in} X}x'{\leq} x\}
$$
has a $\leq$-greatest element $\freshwedge{A}X$. 

Similarly call $\mathcal L$ \deffont{finitely fresh-cocomplete} when for every finite $X$ and $A$ as above the set of $A$-fresh upper bounds 
$\{x'\in|\mathcal L| \mid A\cap\supp(x')=\varnothing\ \wedge\ \Forall{x{\in} X}x{\leq} x'\}$
has a $\leq$-least element $\bigvee^{\#A}X$.
\end{defn}

\begin{lemm}
$\freshwedge{A}X$ and
$\freshvee{A}X$ are unique if they exist.
\end{lemm}
\begin{proof}
From the fact that for a partial order, $x\leq y$ and $y\leq x$ imply $x=y$.
\end{proof}

\begin{nttn}
\label{nttn.Afresh.glb}
Suppose $\mathcal L=(|\mathcal L|,\act,\leq)$ is a finitely fresh-complete and -cocomplete nominal poset. 
Suppose $X\subseteq|\mathcal L|$ and $A\subseteq\atoms$ are finite.
Then:
\begin{itemize*}
\item
Call $\freshwedge{A}X$ an \deffont{$A\#$greatest lower bound} (\emph{`$A$-fresh greatest lower bound'}) of $X$.
If $A=\varnothing$ then we may just call this a greatest lower bound.
If $A=\{a\}$ we may just call this an $a\#$greatest lower bound and write it $\freshwedge{a}X$. 
\item
Call $\freshvee{A}X$ an \deffont{$A\#$least upper bound} of $X$.
If $A=\varnothing$ then we may just call this a least upper bound.
If $A=\{a\}$ we may just call this an $a\#$least upper bound and write it $\freshvee{a}X$. 
\end{itemize*}
\end{nttn}

\begin{nttn}
\label{nttn.lall}
Suppose $\mathcal L=(|\mathcal L|,\act,\leq)$ is a nominal poset.
\begin{itemize*}
\item
Write $\ltop$ for the greatest lower bound and $\lbot$ for the least upper bound of the empty set $\varnothing$, where these exist.
\item
Write $x\land y$ for the greatest lower bound and $x\lor y$ for the least upper bound of $\{x,y\}$, where these exist. 
\item
Write $\lall{a}x$ for the $a\#$greatest lower bound and $\lexi{a}x$ for the $a\#$least upper bound of $\{x\}$, where these exist.
\end{itemize*}
\end{nttn}

\begin{defn}
Suppose $\mathcal L$ is a partial order (it does not matter whether it is nominal).
Call $x'\in|\mathcal L|$ a \deffont{complement} of $x\in|\mathcal L|$ when $x\land x'=\lbot$ and $x\lor x'=\ltop$.
If every $x\in|\mathcal L|$ has a complement say that $\mathcal L$ \deffont{has complements}, and write the complement of $x$ as $\lneg x$.
\end{defn}

\begin{lemm}
Complements are unique if they exist.
\end{lemm}

\begin{corr}
Suppose $\mathcal L=(|\mathcal L|,\act,\leq)$ is a nominal poset.
\begin{enumerate*}
\item
Suppose $X{\subseteq}|\mathcal L|$ is finite and $A{\subseteq}\mathbb A$, and $\freshwedge{A} X$ exists. 

Then $\supp(\freshwedge{A} X){\subseteq} \bigcup\{\supp(x)\mid x\in X\}{\setminus} A$.
\item
Suppose $x{\in}|\mathcal L|$, and suppose $\lneg x$ exists.
Then $\supp(\lneg x)=\supp(x)$.
\end{enumerate*}
\end{corr}
\begin{proof}
We can use symmetry (equivariance) properties of atoms:
\begin{enumerate*}
\item
Using \cite[Theorems~2.29 and~4.7]{gabbay:fountl} $\supp(\freshwedge{A}X)\subseteq\bigcup\{\supp(x)\mid x\in X\}\cup A$.\footnote{This is just a fancy way of observing that since the definition of $\bigwedge$ is symmetric in atoms, the the result must be at least as symmetric as the inputs.} 
Since by assumption $A\cap\supp(\freshwedge{A}X)=\varnothing$, the result follows.
\item
By \cite[Theorem~4.7]{gabbay:fountl} and since 
the map $x\mapsto\lneg x$ is injective.
\end{enumerate*}
\end{proof}

\subsection{Substitution algebra}
\label{subsect.sigma}

The reader may be used to seeing substitution as a concrete operation on sets.
However, using nominal sets we can powerfully (and axiomatically) generalise substitution to be an abstract (nominal) algebraic structure.

\maketab{tab0}{@{\hspace{4em}}L{6em}@{\ }L{3em}@{\ }R{8em}@{\ }L{10em}}

\begin{figure}[t]
\begin{minipage}{\textwidth}
\begin{tab0}
\rulefont{Suba} && a[a{\sm}u]=&u
\\[.5ex]
\rulefont{Subid} && x[a{\sm}a]=&x
\\
\rulefont{Sub\#} &a\#x\limp& x[a{\sm}u]=&x
\\
\rulefont{Sub\alpha}&b\#x\limp&x[a{\sm}u]=&((b\;a)\act x)[b{\sm}u]
\\
\rulefont{Sub\sigma} &a\#v\limp & x[a{\sm}u][b{\sm}v]=&x[b{\sm}v][a{\sm}u[b{\sm}v]]
\end{tab0}
\end{minipage}
\caption{Nominal algebra axioms for substitution action} 
\label{fig.nom.sigma}
\end{figure}

\begin{frametxt}
\begin{defn}
\label{defn.term.sub.alg}
A \deffont{termlike substitution algebra} over $\Sigma$ is a tuple $\ns U=(|\ns U|,\act,\f{atm},\f{sub})$ where:
\begin{itemize*}
\item
$(|\ns U|,\act)$ is a nominal set. 
\item
$\f{atm}:\mathbb A\to|\ns U|$ is an equivariant injection, usually written invisibly (so we write $\f{atm}(a)$ just as $a$). 
\item
$\f{sub}:|\ns U|\times\mathbb A\times|\ns U|\to |\ns U|$ is an equivariant \deffont{substitution action}, written infix $v[a\sm u]$.
\end{itemize*}
such that the equalities \rulefont{Suba} to \rulefont{Sub\sigma} of Figure~\ref{fig.nom.sigma} hold, where $x$, $u$, and $v$ range over elements of $|\ns U|$.
\end{defn}
\end{frametxt}

\begin{frametxt}
\begin{defn}
\label{defn.eq.nu}
Suppose $\ns U=(|\ns U|,\act,\f{sub},\f{atm})$ is a termlike substitution algebra.
A \deffont{substitution algebra} over $\ns U$ is a tuple $\ns B=(|\ns B|,\act,\f{sub})$ where:
\begin{itemize*}
\item
$(|\ns B|,\act)$ is a nonempty nominal set.
\item
$\f{sub}:|\ns B|\times\mathbb A\times|\ns U|\to\ns B$
is an equivariant \deffont{substitution} function satisfying the axioms \rulefont{Subid} to \rulefont{Sub\sigma} of Figure~\ref{fig.nom.sigma}, where $x$ ranges over elements of $|\ns B|$ and $u$ and $v$ range over elements of $|\ns U|$.
\end{itemize*}
\end{defn}
\end{frametxt}

\begin{xmpl}
\label{xmpl.termlike}
\begin{enumerate*}
\item
The set of atoms $\mathbb A$ is a termlike substitution algebra where $\f{atm}(a)=a$ and $a[a{\sm}x]=x$ and $b[a{\sm}x]=b$.
\item
Terms from Definition~\ref{defn.terms} are a termlike substitution algebra with $\f{atm}(a)=a$ and $r[a{\sm}s]$ equal to $r[a\ssm s]$.
\item
Predicates from Definition~\ref{defn.terms} are not a termlike substitution algebra, because there are no predicate variables or substitution for predicates.
Predicates are however a substitution algebra over terms.
\end{enumerate*}
In Corollary~\ref{corr.termlike.sub.alg} we will prove that $|\mathcal N|^\bullet$ from Definition~\ref{defn.termlike.sub.alg} is a termlike substitution algebra, and $X^\bullet$ from Definition~\ref{defn.termlike.sub.alg} is a substitution algebra.
For different and non-trivial classes of substitution algebras see \cite{gabbay:stusun,gabbay:stodfo}.
\end{xmpl}

\section{The model}
\label{sect.model}

\subsection{Nominal Boolean algebra}

Definition~\ref{defn.nom.bool.alg} is the culminating definition of this paper; we define a notion of Boolean algebraic structure suitable for giving an `absolute' interpretation of first-order logic.
This means that variables, substitution, and quantification must be part of the model just as much as conjunction and negation are part of ordinary Boolean algebras.
With our definitions so far, this is now not difficult.

Definition~\ref{defn.nom.bool.alg} is not quite enough to interpret first-order logic on its own.
For that, we need the notion of model which we develop next in Subsection~\ref{subsect.soundness}.

\begin{frametxt}
\begin{defn}
\label{defn.nom.bool.alg}
Suppose $\ns U$ is a termlike substitution algebra. 
A \deffont{nominal Boolean algebra} over $\ns U$ is a tuple $\ns B=(|\ns B|,\act,\leq,\f{sub})$ such that:
\begin{itemize*}
\item
$(|\ns B|,\act,\leq)$ is a finitely fresh-complete nominal poset with complements.
\item
$(|\ns B|,\act,\f{sub})$ is a substitution algebra over $\ns U$.
\item
The partial order is \emph{compatible} with the substitution action: 
$$
\hspace{-2em}\begin{array}{l@{\quad}r@{\ }l}
\rulefont{Compat\bigwedge} 
&
(\freshwedge{A}X)[a\sm u]=&\freshwedge{A}\{x[a\sm u]\mid x{\in} X\}\quad\text{if}\ A{\cap}\supp(u){=}\varnothing,\ a{\not\in} A
\\
\rulefont{Compat\lneg} 
&
(\lneg x)[a\sm u]=&\lneg (x[a\sm u]) .
\end{array}
$$
\end{itemize*}
\end{defn}
\end{frametxt}
A more accurate name for `nominal Boolean algebra' might be `nominal $\sigma\freshwedge{A}\lneg$ partial order'.
But that is a bit of a mouthful.\footnote{A possibly nice classification is to define a notion of \emph{bounded nominal lattice} which is equipped with $\land$, $\lor$, $\lbot$, $\ltop$---and also with $\lall{a}x$.  Then our nominal Boolean algebra is just a nominal lattice with complements and a compatible substitution action.}
It is quite easy to derive some basic properties of nominal Boolean algebras, which are useful for Theorem~\ref{thrm.sound}.

\begin{lemm}
\label{lemm.lub.lall}
If $b\#u$ then $(\lall{b}x)[a\sm u]=\lall{b}(x[a\sm u])$. 
\end{lemm}
\begin{proof}
This is just \rulefont{Compat\bigwedge} where $A=\{b\}$.
Recall that by our permutative convention $a\not\in\{b\}$.
\end{proof}

\begin{lemm}
\label{lemm.leq.glb}
$x\leq y$ if and only if $x\land y=x$.
\end{lemm}
\begin{proof}
Both hold if and only if $x$ is a greatest lower bound for $\{x,y\}$.
\end{proof}

\begin{lemm}
\label{lemm.lub.comm}
\begin{enumerate*}
\item
$(x\land y)[a\sm u]=x[a\sm u]\land(y[a\sm u])$.
\item
$\lbot[a\sm u]=\lbot$.
\end{enumerate*}
\end{lemm}
\begin{proof}
From \rulefont{Compat\bigwedge} and \rulefont{Compat\lneg} respectively.
\end{proof}

\begin{corr}
\label{corr.leq.compat.sub}
If $x\leq y$ then $x[a\sm u]\leq y[a\sm u]$.
As a further corollary, if $x\leq y$ and $a\#x$ then $x\leq y[a\sm u]$ for every $u\in|\ns U|$.
\end{corr}
\begin{proof}
This follows from Lemma~\ref{lemm.leq.glb} and part~1 of Lemma~\ref{lemm.lub.comm}.
The further corollary follows using \rulefont{Sub\#} from Figure~\ref{fig.nom.sigma}.
\end{proof}

\begin{lemm}
\label{lemm.tallL.sound}
Suppose $x\in|\ns B|$ and $u\in|\ns U|$ and $a\in\atoms$.
Then $\lall{a}x\leq x[a\sm u]$.
\end{lemm}
\begin{proof}
By Notation~\ref{nttn.lall} $\lall{a}x$ is the $a\#$greatest lower bound for $\{x\}$.
This means that $\lall{a}x\leq x$ and $a\#\lall{a}x$.
We use Corollary~\ref{corr.leq.compat.sub}.
\end{proof}

\begin{lemm}
\label{lemm.tallR.sound}
Suppose $x,y{\in}|\ns B|$ and $a{\in}\atoms$. 
Then $x\leq y$ and $a\#x$ imply $x\leq\lall{a}y$.
\end{lemm}
\begin{proof}
This is the $a\#$greatest lower bound property of $\lall{a}y$ with respect to $\{y\}$. 
\end{proof}

\begin{corr}
$\lall{a}x$ is the greatest lower bound of $\{x[a\sm u]\mid u\in|\ns U|\}$ in $\ns B$.
\end{corr}
\begin{proof}
By Lemma~\ref{lemm.tallL.sound} $\lall{a}x$ is a lower bound.
Now suppose $z\leq x[a\sm u]$ for every $u\in|\ns U|$.
In particular then, $z\leq x[a\sm a]\stackrel{\rulefont{Subid}}{=}x$.
By Lemma~\ref{lemm.tallR.sound} $z\leq\lall{a}x$.
\end{proof}
 
\subsection{Soundness}
\label{subsect.soundness}

\begin{frametxt}
\begin{defn}
\label{defn.model}
A \deffont{model} $\interp M=(\ns U,\ns B,\text{-}^\mden)$ consists of: 
\begin{itemize*}
\item
A termlike substitution algebra $\ns U$ (Definition~\ref{defn.term.sub.alg}).
\item
A nominal Boolean algebra $\ns B$ over $\ns U$ (Definition~\ref{defn.nom.bool.alg}).
\item
An assignment $\text{-}^\mden$ to each term-former $\tf f$ of an equivariant function $\tf f^\mden:|\ns U|^{\f{arity}(\tf f)}\to |\ns U|$ and to each predicate-former $\tf P$ of an equivariant function $\tf P^\mden:|\ns U|^{\f{arity}(\tf P)}\to |\ns B|$ such that
$$
\begin{array}{@{\hspace{-2em}}l@{\quad}r@{\ }l}
\rulefont{Comm\tf f} & 
\tf f^\mden(u_1,\dots,u_{\f{arity}(\tf f)})[a\sm u]=&
\tf f^\mden(u_1[a\sm u],\dots,u_{\f{arity}(\tf f)}[a\sm u])
\\
\rulefont{Comm\tf P} & 
\tf P^\mden(u_1,\dots,u_{\f{arity}(\tf f)})[a\sm u]=&
\tf P^\mden(u_1[a\sm u],\dots,u_{\f{arity}(\tf f)}[a\sm u]).
\end{array}
$$
\end{itemize*}
\end{defn}
\end{frametxt}

For the rest of this section fix a model $\mathcal M=(\ns U,\ns B,\text{-}^\mden)$.

\begin{defn}
Define an \deffont{interpretation} function $\model{\text{-}}$ mapping terms and predicates (Definition~\ref{defn.terms}) to elements of $|\ns U|$ and $|\ns B|$ respectively, as follows:
$$
\begin{array}{r@{\ }l@{\qquad}r@{\ }l}
\model{a}=&\f{atm}(a)
&
\model{\tf f(r_1,\dots,r_n)}=&\tf f^\mden(\model{r_1},\dots,\model{r_n})
\\
\model{\tbot}=&\lbot
&
\model{\tf P(r_1,\dots,r_n)}=&\tf P^\mden(\model{r_1},\dots,\model{r_n})
\\
\model{\phi'\tand\phi}=&\model{\phi'}\land\model{\phi}
&
\model{\tall a.\phi}=&\lall{a}\model{\phi}
\\
\model{\tneg \phi}=&\lneg\model{\phi}
\end{array}
$$
\end{defn}

\begin{defn}
Write $\interp M\ment \phi$ when $\model{\phi}=\ltop$ and call $\phi$ \deffont{valid} in $\interp M$.
\end{defn}

\begin{prop}
\label{prop:substitution}
$\model{r[a\ssm s]}=\model{r}[a\sm \model{s}]$ and $\model{\phi[a\ssm s]}=\model{\phi}[a\sm\model{s}]$.
\end{prop}
\begin{proof}
By routine inductions on $r$ and $\phi$:
\begin{itemize*}
\item
\emph{The case of $a$.}\quad
$\model{a[a\ssm s]}=\model{s}$ and by \rulefont{Subid} $\model{a}[a\sm\model{s}]=\model{s}$.
\item
\emph{The case of $\tf f(r_1,\dots,r_n)$.}\quad
Using \rulefont{Comm\tf f}. 
\item
\emph{The case of $\phi'\tand\phi$.}\quad
Using part~1 of Lemma~\ref{lemm.lub.comm}.
\item
\emph{The case of $\tall a.\phi$.}\quad
Using Lemma~\ref{lemm.lub.lall}.
\end{itemize*} 
\end{proof}

\begin{thrm}
\label{thrm.sound}
If $\Phi\cent\Psi$ then $\bigwedge_{\phi\in\Phi}\model{\phi}\leq\bigvee_{\hspace{-.6ex}\psi\in\Psi}\model{\psi}$. 
\end{thrm}
\begin{proof}
\begin{itemize*}
\item
Soundness of \rulefont{\tand L} and \rulefont{\tand R} follows from $\model{\phi_1}\land\model{\phi_2}$ being a greatest $\leq$-lower bound of $\model{\phi_1}$ and $\model{\phi_2}$. 
\item
Soundness of \rulefont{\tneg L} and \rulefont{\tneg R} follows from the fact of complements that $\lneg x\leq y$ if and only if $\lneg y\leq x$.
\item
Soundness of \rulefont{\tbot L} follows from $\lbot$ being the complement of the greatest lower bound of $\varnothing$.
\item
Soundness of \rulefont{\tall L} follows using Lemma~\ref{lemm.tallL.sound} and Proposition~\ref{prop:substitution}, and soundness of \rulefont{\tall R} follows from Lemma~\ref{lemm.tallR.sound}. 
\end{itemize*}
\end{proof}

\section{Lifting ordinary models to nominal models}
\label{sect.complete}

We have to recall the (standard) definition of model for first-oder logic.
This is Subsection~\ref{subsect.standard}.
Then in Subsection~\ref{subsect.nom.model.from.standard} we show how to lift this to a nominal model and obtain completeness as an immediate corollary (Theorem~\ref{thrm.functional.model} and Corollary~\ref{corr.nom.complete}).

\subsection{Ordinary model of first-order logic}
\label{subsect.standard}

We briefly sketch the ordinary model of first-order classical logic, with valuations and without atoms.
This model is not intended to be sophisticated; we will just need that one exists.

\begin{nttn}
To avoid the confusion between sets and nominal sets, we may write \emph{ordinary set} for the former.
\end{nttn}

\begin{defn}
\label{defn.model.standard}
An \deffont{ordinary model} $\mathcal N=(|\mathcal N|,\text{-}^\nden)$ is a tuple such that:
\begin{itemize*}
\item
$|\mathcal N|$ is some non-empty (ordinary) \deffont{underlying set}.
\item
$\text{-}^\nden$ assigns to each term-former $\tf f$ a function $\tf f^\nden:|\mathcal N|^{\f{arity}(\tf f)}\to |\mathcal N|$ and to each predicate-former $\tf P$ a function $\tf P^\nden:|\mathcal N|^{\f{arity}(\tf P)}\to \{\lbot,\ltop\}$. 
\end{itemize*}
\end{defn}

\begin{defn}
\label{defn.varsigma}
Write $\atoms\Func |\mathcal N|$ for the set of functions from atoms to $|\mathcal N|$.
Let $\varsigma$ range over elements of $\atoms\Func |\mathcal N|$ and call these \deffont{valuations} (to $\mathcal N$).

Also if $x\in|\mathcal N|$ then write $\varsigma[a\ssm x]$ for the function such that $(\varsigma[a\ssm x])(a)=x$ and $(\varsigma[a\ssm x])(b)=\varsigma(b)$.
\end{defn}

\begin{defn}
$\{\lbot,\ltop\}$ is a complete Boolean algebra.
We use standard definitions, such as $\land$, $\lneg$, $\bigwedge$, and $\bigvee$ without further comment
\end{defn}

\begin{defn}
Define an \deffont{interpretation} function $\ndenot{\varsigma}{\text{-}}$ mapping terms and predicates (Definition~\ref{defn.terms}) to elements of $|\mathcal N|$ and $\{\lbot,\ltop\}$ (considered as a Boolean algebra) respectively, as follows:
$$
\begin{array}{r@{\ }l@{\qquad}r@{\ }l}
\ndenot{\varsigma}{a}=&\varsigma(a)
&
\ndenot{\varsigma}{\tf f(r_1,\dots,r_n)}=&\tf f^\nden(\ndenot{\varsigma}{r_1},\dots,\ndenot{\varsigma}{r_n})
\\
\ndenot{\varsigma}{\tbot}=&\lbot
&
\ndenot{\varsigma}{\tf P(r_1,\dots,r_n)}=&\tf P^\nden(\ndenot{\varsigma}{r_1},\dots,\ndenot{\varsigma}{r_n})
\\
\ndenot{\varsigma}{\phi'\tand\phi}=&\ndenot{\varsigma}{\phi'}\land\ndenot{\varsigma}{\phi}
&
\ndenot{\varsigma}{\tall a.\phi}=&\bigwedge_{x\in|\mathcal N|}\ndenot{\varsigma[a\ssm x]}{\phi}
\\
\ndenot{\varsigma}{\tneg \phi}=&\lneg\ndenot{\varsigma}{\phi}
\end{array}
$$
\end{defn}

Theorem~\ref{thrm.fol.standard} expresses the usual soundness and completeness result for first-order logic; for details and proofs see e.g. \cite[Subsection~1.5]{dalen:logs}:
\begin{thrm}
\label{thrm.fol.standard}
$\Phi\cent\Psi$ is derivable if and only if for every ordinary model $\mathcal N$ and every valuation $\varsigma$ to $\mathcal N$ it is the case that $\bigwedge_{\phi\in\Phi}\ndenot{\varsigma}{\phi}=\ltop$ implies $\bigvee_{\hspace{-.6ex}\psi\in\Psi}\ndenot{\varsigma}{\psi}=\ltop$.
\end{thrm}

\subsection{Constructing a nominal model from an ordinary model}
\label{subsect.nom.model.from.standard}

We now show how to `lift' a model over ordinary sets to a nominal model (Proposition~\ref{prop.standard.nom.bool}).
We then deduce completeness for nominal Boolean algebras (Corollary~\ref{corr.nom.complete}).

\begin{defn}
\label{defn.varsigma.action}
Give valuations $\varsigma\in\atoms\Func|\mathcal N|$ a permutation action by
$$
(\pi\act\varsigma)(a)=\varsigma(\pi^\mone(a)) .
$$
\end{defn}

\begin{rmrk}
\label{rmrk.not.nominal}
$\atoms\Func |\mathcal N|$ forms a set with a permutation action (Definition~\ref{defn.fin.supp}).
This is not in general a nominal set because it has elements without finite support\footnote{The finitely-supported $\varsigma$ are such that there exists a finite $A\subseteq\mathbb A$ such that for all $a,b\not\in A$, $\varsigma(a)=\varsigma(b)$; it is not worth our while to impose this restriction, though it would do no harm to do so.} 
for our purposes that will not be a problem.

In fact, the action from Definition~\ref{defn.varsigma.action} is a special case of the standard \emph{conjugation action} $(\pi\act\varsigma)(a)=\pi\act\varsigma(\pi^\mone(a))$ where we use the trivial action $\pi\act x=x$ for every $x\in |\mathcal N|$.
This is standard; for a specifically `nominal' discussion see \cite[Definition~2.4.2]{gabbay:nomtnl}.
\end{rmrk}

\begin{defn}
\label{defn.idenot.beta}
Given an ordinary set $X$ write $X^\bullet=(|X^\bullet|,\act)$ for 
\begin{itemize*}
\item
$|X^\bullet|$ is the set of functions $f$ from $\atoms\Func |\mathcal N|$ to $X$ such that there exists a finite set $A$ such that if $\varsigma(a)=\varsigma'(a)$ for every $a\in A$ then $f(\varsigma)=f(\varsigma')$, and
\item
the permutation action $\act$ is defined by 
$$
(\pi\act f)(\varsigma)=f(\pi^\mone\act\varsigma).
$$
\end{itemize*}
We will be most interested in $|\mathcal N|^\bullet$ and $\{\lbot,\ltop\}^\bullet$.
\end{defn}

\begin{lemm}
\label{lemm.idenot.beta.pnom}
For any ordinary set $X$,\ $X^\bullet$ from Definition~\ref{defn.idenot.beta} determines a nominal set (Definition~\ref{defn.nominal.set}).
\end{lemm}
\begin{proof}
It is routine to verify that the permutation action is indeed a group action.
It remains to check finite support.

Suppose $\pi(a)=a$ for every $a\in A$ (notation from Definition~\ref{defn.idenot.beta}).
By Definition~\ref{defn.idenot.beta} $(\pi\act f)(\varsigma)=f(\pi^\mone\act\varsigma)$.
By Definition~\ref{defn.varsigma.action} $(\pi^\mone\act\varsigma)(a)=\varsigma(\pi(a))$.
Now by assumption $\varsigma(a)=\varsigma(\pi(a))$ for every $a\in A$.
Therefore, $(\pi^\mone\act\varsigma)(a)=\varsigma(a)$ for every $a\in A$, and so $f(\varsigma)=f(\pi^\mone\act\varsigma)$, and so $(\pi\act f)(\varsigma)=f(\varsigma)$. 
Thus, $f$ has finite support (and is supported by $A$).
\end{proof}

\begin{defn}
\label{defn.termlike.sub.alg}
Suppose $\mathcal N$ is an ordinary model (Definition~\ref{defn.model.standard}).
Specify a tuple $(||\mathcal N|^\bullet|,\act,\f{atm},\f{sub})$ by: 
\begin{itemize*}
\item
$\f{atm}$ is defined by 
$\f{atm}(a)(\varsigma)=\varsigma(a)$.
\item
$\f{sub}$ is defined by
$f[a\sm g](\varsigma)=f(\varsigma[a\ssm g(\varsigma)])$. 
\end{itemize*}
We may write $|\mathcal N|^\bullet$ for $(||\mathcal N|^\bullet|,\act,\f{atm},\f{sub})$.\footnote{So we abuse notation by overloading with the nominal set, and $||\mathcal N|^\bullet|$ is: the underlying set of the substitution algebra obtained from the underlying set of the ordinary model $\mathcal N$. What could be simpler?} 

Furthermore if $X$ is any ordinary set then (abuse notation again and) write $X^\bullet=(|X^\bullet|,\act,\f{sub})$ where $f[a\sm g](\varsigma)=f(\varsigma[a\ssm g(\varsigma)])$ for $f\in|X^\bullet|$ and $g\in||\mathcal N|^\bullet|$.
\end{defn}

\begin{lemm}
\label{lemm.supp'.supp}
Suppose $X$ is an ordinary set and $f\in |X^\bullet|$.
Then if $\varsigma(a)=\varsigma'(a)$ for every $a\in\supp(f)$ then $f(\varsigma)=f(\varsigma')$. 
\end{lemm}
\begin{proof}
Suppose $\varsigma(a)=\varsigma'(a)$ for every $a\in\supp(f)$.
It suffices to show that if $a\in A\setminus\supp(f)$ and $x\in |\mathcal N|$ then $f(\varsigma[a\ssm x])=f(\varsigma)$.

Choose fresh $b:\alpha$ (so $b\not\in A$).
By part~1 of Corollary~\ref{corr.stuff} $(b\ a)\act f=f$, since $a,b\not\in\supp(f)$.
We reason as follows:
$$
f(\varsigma[a\ssm x])\stackrel{b{\not\in}A}{=} f(\varsigma[a\ssm x][b\ssm\varsigma(a)])\stackrel{(b\,a)\act f{=}f}{=}f(\varsigma[b\ssm x])\stackrel{b{\not\in}A}{=} f(\varsigma)
$$ 
\end{proof}

\begin{corr}
\label{corr.termlike.sub.alg}
$\mathcal N^\bullet$ from Definition~\ref{defn.termlike.sub.alg} is indeed a termlike substitution algebra, and $X^\bullet$ is indeed a substitution algebra over $\mathcal N^\bullet$.
\end{corr}
\begin{proof}
We examine Definitions~\ref{defn.term.sub.alg} and~\ref{defn.eq.nu} and see that we need to check equivariance and the axioms \rulefont{Suba} (for the termlike case) and \rulefont{Subid} to \rulefont{Sub\sigma} from Figure~\ref{fig.nom.sigma}.
This is routine; we use Lemma~\ref{lemm.supp'.supp}.
\end{proof}

\begin{defn}
\label{defn.leq.nsN}
Suppose $g',g\in |\{\lbot,\ltop\}^\bullet|$. 
Write $g'\leq g$ when $\Forall{\varsigma}g'(\varsigma)\leq g(\varsigma)$.
\end{defn}

\begin{lemm}
\label{lemm.bot.top.ffcl}
$(|\{\lbot,\ltop\}^\bullet|,\act,\leq)$ is a finitely fresh-complete lattice with complements.
\end{lemm}
\begin{proof}
It is easy to check that $\lneg g$ defined by $(\lneg g)(\varsigma)=\lneg(g(\varsigma))$ is a complement to $g$.

Now suppose $X\subseteq|\{\lbot,\ltop\}^\bullet|$ and $A\subseteq \atoms$ are finite, and suppose $A\cap\supp(u)=\varnothing$ and $a\in\atoms{\setminus} A$.
Define 
$$
(\freshwedge{A}X)(\varsigma)=\bigwedge\{g(\varsigma[a\ssm x_a]_{a\in A})\mid g\in X,\ \Forall{a{\in}A}x_a\in|\mathcal N|\} .
$$
It is routine to verify that this is an $A\#$greatest lower bound for $X$.
\end{proof}

\begin{prop}
\label{prop.standard.nom.bool}
$(|\{\lbot,\ltop\}^\bullet|,\act,\leq,\f{sub})$ is a nominal Boolean algebra over $|\mathcal N|^\bullet$.
\end{prop}
\begin{proof}
We unpack Definition~\ref{defn.nom.bool.alg}.
\begin{itemize*}
\item
$(|\{\lbot,\ltop\}^\bullet|,\act)$ is a nominal set by Lemma~\ref{lemm.idenot.beta.pnom}.
\item
$(|\{\lbot,\ltop\}^\bullet|,\act,\leq)$ is by Lemma~\ref{lemm.bot.top.ffcl} a finitely fresh-complete lattice with complements.
\item
\rulefont{Compat\bigwedge} and \rulefont{Compat\lneg} follow by easy calculations, since Definition~\ref{defn.leq.nsN} is pointwise on valuations. 
\end{itemize*}
\end{proof}

\begin{defn}
\label{defn.functional.model}
The interpretation of term-formers and predicate-formers is simple:
$$
\begin{array}{r@{\ }l}
\tf f^\iden(f_1,\dots,f_{\f{arity}(\tf f)})(\varsigma)=&\tf f^\nden(f_1(\varsigma),\dots,f_{\f{arity}(\tf f)}(\varsigma))
\\
\tf P^\iden(f_1,\dots,f_{\f{arity}(\tf f)})(\varsigma)=&\tf P^\nden(f_1(\varsigma),\dots,f_{\f{arity}(\tf P)}(\varsigma))
\end{array}
$$
\end{defn}

\begin{thrm}
\label{thrm.functional.model}
$(|\mathcal N|^\bullet,\{\lbot,\ltop\}^\bullet,\text{-}^\iden)$ 
determines a model in the sense of Def.~\ref{defn.model}.
\end{thrm}
\begin{proof}
$|\mathcal N|^\bullet$ is a termlike substitution algebra by Corollary~\ref{corr.termlike.sub.alg}.
$\{\lbot,\ltop\}^\bullet$ is a nominal Boolean algebra over $|\mathcal N|^\bullet$ by Proposition~\ref{prop.standard.nom.bool}.
The commutation conditions \rulefont{Comm\tf f} and \rulefont{Comm\tf P} are by easy calculations.
\end{proof}

\begin{corr}
\label{corr.nom.complete}
If $\Phi\not\cent\Psi$, then there exists a termlike substitution algebra $\ns U$ and a nominal Boolean algebra $\ns M$ over $\ns U$ such that $\bigwedge_{\phi\in\Phi}\model{\phi}\not\leq\bigvee_{\psi\in\Psi}\model{\psi}$.

As a corollary, if $\Phi$ is inconsistent (that is, if $\Phi\cent\lbot$) then $\Phi$ has no model.
\end{corr}
\begin{proof}
By completeness of first-order logic there exists some ordinary model $\mathcal N$ (Definition~\ref{defn.model.standard}) and valuation $\varsigma$ to $\mathcal N$ such that 
$\bigwedge_{\phi\in\Phi}\ndenot{\varsigma}{\phi}=\ltop$ and $\bigvee_{\hspace{-.6ex}\psi\in\Psi}\ndenot{\varsigma}{\psi}=\lbot$.
It follows from Definition~\ref{defn.leq.nsN} that 
$\bigwedge_{\phi\in\Phi}\model{\phi}\not\leq\bigvee_{\hspace{-.6ex}\psi\in\Psi}\model{\psi}$.
\end{proof} 

It is also possible give a direct proof of completeness of first-order logic with respect to the models of Definition~\ref{defn.model}, using a model based on syntax quotiented by derivable equality.
The advantage of the proof of completeness used here is that it also shows that completeness does \emph{not} hold just because the only models available are syntax-based.
Topological models \cite{gabbay:stodfo} and set-theoretic models \cite{gabbay:stusun} also exist (those papers predate the consideration of lattices and $A\#$fresh limits of this paper).

\section{Conclusions}

The usual Tarski semantics, based on valuations, has names in the syntax but not in the semantics.
Names can be used to name semantic objects but these objects do not themselves contain names.
We have challenged the prejudice that only syntax contains names, and proposed a framework where names are (also) in denotation.

This is related to a broad debate in philosophy about to what extent names are real. 
One investigation in that literature deserves special mention: Kit Fine's notion of \emph{arbitrary object} \cite{fine:reaao}, which corresponds to our notion of `names in denotation' and in particular to our notion of nominal substitution set, though the technical details of Fine's constructions are quite different from ours.\footnote{Fine did not have permutations or finite support independently of and preceding the notion of substitution.}

This paper is part of a broader research context.
For instance in \cite{gabbay:stodfo} we built a representation theorem for first-order logic (in retrospect this also provides a class of `topological' models for the axioms in this paper, distinct from the `lifted' models of Section~\ref{sect.complete}).
In \cite{gabbay:stusun} we show that, remarkably, a large subclass of the cumulative hierarchy model of Fraenkel-Mostowski sets forms a substitution algebra.

In this paper we apply the technology to 
the theory of lattices, in particular identifying $A\#$fresh limits as the single unifying notion of greatest lower bound needed to model both conjunction and universal quantification.
This is another step towards a `nominal' view of logic and computation whereby names, and associated structures, are treated as integral to denotation.

\subsubsection*{Related and future work}

A nominal style semantics related to the semantics of this paper has been given for the $\lambda$-calculus \cite{gabbay:nomhss}.
Nominal lattices have not been considered to date.
The author and collaborators considered a notion of \emph{Banona} in \cite{gabbay:stodnb} (Boolean algebra with $\new$) which is tangentially related to this paper.
Domain theory in nominal sets is considered in \cite{turner:thesis} and a notion of \emph{FM category} is considered in \cite{clouston:nomlt,clouston:equlnb}.
The notion of $A\#$fresh limit seems to be new though, as is the specifically lattice-theoretic treatment.

A relative of nominal lattices is cylindric set algebras \cite{monk:intcsa,henkin:cyla}, a semantics expressed in terms of relations between elements of the domain of the model, where quantification is interpreted using projection and cylindrification. 
Instead of designing semantics for first-order logic based on relations, projection, and cylindrification, we have taken the novel approach of giving a direct interpretation to variables in the model.
In fact, the `ordinary models' we build in Subsection~\ref{subsect.standard} correspond to locally finite dimensional regular cylindric set algebras \cite{monk:intcsa}.
There is no non-nominal precedent to the nominal Boolean algebras from Definition~\ref{defn.nom.bool.alg}.

There are two natural avenues for future work:

\paragraph{Add binders to the term-language.}

Part of the challenge of this is to develop a suitable new syntax extending first-order logic with term-formers that can bind. 
Once we have this, it is easy to interpret binders, because our (nominal) denotation supports atoms-abstraction as primitive.\footnote{Atoms-abstraction $[a]x$ is the native (non-functional) notion of abstraction in nominal techniques.  It was introduced in \cite{gabbay:newaas-jv}.}

In particular, we see this being applied computationally in formalising mathematics, where binders are everywhere, as an improvement over first-order logic (where term-formers cannot bind) or higher-order logic (which may already be unnecessarily powerful).

This would be related to the \emph{permissive-nominal logic} of \cite{gabbay:pernl,gabbay:pernl-jv} except that here, we would want to take substitution as primitive.
This is related to work by Fiore and his students \cite{fiore:abssvb,fiore:teresl} which tries to develop (purely in category theory) a unified account of names, binding, and substitution; a non-exhaustive but accessible introduction to these ideas is in \cite{power:absssb}. 
Slightly further afield, Beeson's \emph{lambda-logic} \cite{beeson:laml} enriches first-order logic with a term-language that is specifically the untyped $\lambda$-calculus; the denotational assumptions used are completely different, but the spirit of allowing binding term-formers in a first-order logic is similar. 
Of course this also goes back to the first author's \emph{Binding Logic} \cite{dowek:binl}.

\paragraph{Generalise to categories.}

Another obvious next step is to generalise the constructions in this paper from lattices to categories, and so obtain an enriched category theory.
By this view, we note that lattices are just a special case of categories, and our $A\#$fresh limits are just limits in a subcategory of elements of the lattice for which $A$ is fresh.
Much can be done here.
For instance, we can try to identify the correct common generalisation of this new view of quantifiers and the standard one based on left- and right-adjoints to projection \cite[Section~5.5]{PittsAM:catl}.

Slightly less ambitiously, we can simply consider a general notion of \emph{nominal lattice}, which enriches ordinary lattices with $\lall{a}x$ (a greatest $a\#$lower bound for $x$) but without necessarily assuming a substitution action.\footnote{The $\hnu a.x$ of \cite{gabbay:stodnb} is not an instance of this because it commutes with negation, but the (dual of) $\nu a.x$ from \cite{gabbay:frenrs} is an instance of this. There are two distinct `restriction' operations here.}
Clearly, there is a rich theory of nominal preorders which remains to be explored.

\vspace{-1em}

\noindent {\footnotesize \emph{Acknowledgements.}
The second author acknowledges the support of the Leverhulme trust.
}

\end{document}